\documentclass[parindent=0pt]{article}
\usepackage{booktabs}
\usepackage{makecell}
\usepackage{subcaption}

\usepackage[utf8]{inputenc}     
\usepackage[english]{babel}     
\usepackage[a4paper, left=1in, right=1in, top=1.25in, bottom=1.25in]{geometry}
\usepackage{graphicx}           

\usepackage{amsmath,amssymb}    
\usepackage{amsthm}             
\usepackage{mathtools}          
\usepackage{enumitem}           
\usepackage{listings}           
\usepackage{todonotes}          
\usepackage{hyperref}           
\usepackage{float}
\usepackage{amsthm}
\usepackage{amsmath}
\usepackage{amscd}
\usepackage{amssymb}
\usepackage{mathrsfs}
\usepackage{pdfpages}
\usepackage{setspace}
\usepackage{physics}
\usepackage{tcolorbox}

\usepackage{xcolor}
\usepackage{tikz}
\usepackage[toc,page,title,titletoc,header]{appendix}
 \usepackage{url}
 \usepackage{babel} 
\usepackage{csquotes}
\usepackage[f]{esvect}
\usepackage{blindtext}

\newcommand{\R}{\mathbb{R}}

\newcommand{\E}{\mathbb{E}}
\newcommand{\C}{\mathbb{C}}

\newcommand{\g}{\mathfrak{g}}
\newcommand{\h}{\mathfrak{h}}
\newcommand{\s}{\mathfrak{s}}

\newcommand{\diag}{\mbox{diag}}

\let\phi=\varphi
\let\epsilon=\varepsilon

\newtheorem{theorem}{Theorem}[section]
\newtheorem{conjecture}{Conjecture}[section]

\newtheorem{lemma}[theorem]{Lemma}

\newcommand{\adots}{\mathinner{\mkern1mu \raise1pt
\hbox{.} \mkern2mu \raise4pt \hbox{.} \mkern2mu \raise7pt
\vbox{\kern7pt \hbox{.}} \mkern1mu}}
\def\[{[\![}
\def\]{]\!]}

\newlist{regimes}{enumerate}{1}
\setlist[regimes]{label={\textbullet\ Regime (\arabic*):}, leftmargin=*}
\usepackage{authblk} 

\title{%
   Distribution of the Diagonal Entries of the Resolvent of a Complex Ginibre Matrix\\
  }

  \author[1,2]{Pierre Bousseyroux\thanks{Email: pierre.bousseyroux@polytechnique.edu}}
\author[3]{Jean-Philippe Bouchaud}
\author[3]{Marc Potters}

\affil[1]{Chair of Econophysics and Complex Systems, Ecole Polytechnique, 91128 Palaiseau, France}
\affil[2]{LadHyX, UMR CNRS 7646, Ecole Polytechnique, Institut Polytechnique de Paris, 91128 Palaiseau, France}
\affil[3]{Capital Fund Management, Paris, France}

\begin{document}
\maketitle

\begin{abstract}
   The study of eigenvalue distributions in random matrix theory is often conducted by analyzing the resolvent matrix $ \vb{G}_{\vb{M}}^N(z) = (z \vb{1} - \vb{M})^{-1} $. The normalized trace of the resolvent, known as the Stieltjes transform $ \g_{\vb{M}}^N(z) $, converges to a limit $ \g_{\vb{M}}(z) $ as the matrix dimension $ N $ grows, which provides the eigenvalue density $ \rho_{\vb{M}} $ in the large-$ N $ limit. In the Hermitian case, the distribution of $ \g_{\vb{M}}^N(z) $, now regarded as a random variable, is explicitly known when $ z $ lies within the limiting spectrum, and it coincides with the distribution of any diagonal entry of $\vb{G}_{\vb{M}}^N(z)$.

In this paper, we investigate what becomes of these results when $ \vb{M} $ is non-Hermitian. Our main result is the exact computation of the diagonal elements of $\vb{G}_{\vb{M}}^N(z)$ when $ \vb{M} $ is a Ginibre matrix of size $ N $, as well as the high-dimensional limit for different regimes of $ z $, revealing a tail behavior connected to the statistics of the left and right eigenvectors. Interestingly, the limit distribution is stable under inversion, a property previously observed in the symmetric case. We then propose two general conjectures regarding the distribution of the diagonal elements of the resolvent and its normalized trace in the non-Hermitian case, both of which reveal a symmetry under inversion.

\end{abstract}

\section{Introduction}

The standard way to compute the eigenvalue distribution of  a random matrix $\vb{M}$ is to study the resolvant matrix of $\vb{M}$ defined as 
\begin{equation}
    \vb{G}_{\vb{M}}^N(z) = (z\vb{1} - \vb{M})^{-1}
\end{equation} where $z$ is a complex variable defined away from all the eigenvalues of $\vb{M}$ and $\vb{1}$ denotes the identity matrix. Then, the Stieltjes transform of $\vb{M}$ is given by the normalized trace of the resolvent
\begin{equation}
    \g_{\vb{M}}^N(z)= \frac{1}{N}\Tr\left(\vb{G}^N_{\vb{M}}(z)\right) = \frac{1}{N}\sum_{k=1}^N \frac{1}{z - \lambda_k}
\end{equation} where $\lambda_k$ are the eigenvalues of $\vb{M}$. 
One may refer to \cite{bordenave2012around} for a more careful approach to the next assertions and formulas. If the random matrix is correctly normalized, in many instances, the eigenvalue spectrum of large random matrices is confined in a compact set. We denote $\rho_{\vb{M}}$ the limit continuous distribution of the eigenvalues of $\vb{M}$ when $N$ gets large. Then, $\g_{\vb{M}}^N(z)$ often admits a limit denoted $\g_{\vb{M}}$ when $N$ goes to infinity with $z$ outside the continuous spectrum of $\vb{M}$ and does not depend anymore of a single realization of $\vb{M}$:
\begin{equation}
    \g_{\vb{M}}^N(z) \underset{N\to \infty}{\longrightarrow} \int \frac{\rho_{\vb{M}}(\lambda)}{z - \lambda}d\lambda
\end{equation}when $\vb{M}$ is Hermitian and
\begin{equation}
    \int \int \frac{\rho_{\vb{M}}(\lambda)}{z - \lambda} d\Re(\lambda)\wedge d\Im(\lambda)
\end{equation}otherwise. Unlike in the Hermitian case, when $\vb{M}$ is non-Hermitian, $\g_{\vb{M}}(z)$ can be defined even when $z$ lies within the spectrum.

For a Hermitian matrix $\vb{M}$, the computation of $\g_{\vb{M}}$ enables us to determine $\rho_{\vb{M}}$ via the Sokhotski-Plemelj formula:
\begin{equation}
    \underset{\eta \to 0^{+}}{\text{lim}}\Im \g_{\vb{M}}(x - i \eta) = \pi \rho_{\vb{M}}(x)
\end{equation}while when $\vb{M}$ is not Hermitian, one can recover $\rho_{\vb{M}}$ using
\begin{equation}
    \partial_{\overline{z}} \g_{\vb{M}} = \pi \rho_{\vb{M}}.
\end{equation}
For instance, if $\vb{M}$ is a complex Ginibre matrix, it is known that\footnote{We recall that a complex Ginibre ensemble is an $N\times N$ non-Hermitian random matrix over $\C$ with i.i.d. complex Gaussian entries normalized to have mean zero and variance $\frac{1}{N}$.}
\begin{equation}\label{eq:ginibre}
    \g(z) = \left\{
    \begin{array}{cc}
        \overline{z} & \mbox{when} |z|\leq 1 \\
        \frac{1}{z} & \mbox{otherwise.}
    \end{array}
    \right.
\end{equation}and so, $\rho_{\vb{M}}(z) = \frac{1}{\pi}$ for $|z|\leq 1$ and $0$ otherwise. The distribution of complex eigenvalues of $\vb{M}$ is uniform in the unit disk. 

A different way to recover $\rho_{\vb{M}}(z)$ is to consider $\g_{\vb{M}}^{N}(z)$ as a random variable $g$ and study the tail behaviour of the distribution of $g$ when $|g|\to\infty$. Indeed, the tail amplitude of this variable is directly governed by the presence of eigenvalues around $z$. This idea was proposed used in \cite{cizeau1994theory} to compute the eigenvalue distribution of L\'evy matrices. When $\vb{M}$ is a Hermitian matrix and $N$ becomes large, the tail of the distribution of $g$ should behave like ${\rho_{\vb{M}}(x)/}{g^2}$ for $z=x$ real.
In fact, in the case where $\vb{M}$ is an $N\times N$ Wigner matrix, i.e. with independent entries having a mean $0$ and variance $\sigma^2/N$, we even know the full limiting distribution of $\g_{\vb{M}}^N$. It follows a Cauchy distribution with a density function given by \begin{equation}\label{dis} \frac{\rho_{\vb{M}}(x)}{\left(g - \h_{\vb{M}}(x)\right)^2 + \pi^2 \rho_{\vb{M}}(x)^2}, \end{equation} where $\h_{\vb{M}}(x)$ is the Hilbert transform of $\rho_{\vb{M}}$, which simplifies to $\frac{x}{2\sigma^2}$ in the Wigner case. This result was proven by Fyodorov and collaborators for the GOE and GUE ensembles using Random Matrix Theory techniques (\cite{fyodorov1997statistics}, \cite{fyodorov2007replica}, and \cite{fyodorov2004statistics}). Interestingly, this result is super-universal, in the sense that it holds for a much broader class of point processes on the real axis. For a physicist's perspective, see \cite{aizenman2015ubiquity} and \cite{bouchaud2018two}. Note that \eqref{dis} asymptotically behaves as ${\rho_{\vb{M}}(x)}/{g^2}$ for large $|g|$ as expected.
Interestingly, a computation shows that the distribution \eqref{dis} is stable under inversion. Furthermore, one can infer from \cite{aizenman2015ubiquity} that
$\left[\vb{G}_{\vb{M}}^{N}\right]_{11}$ and $\g_{\vb{M}}^{N}$ have the same limiting distribution.

We then ask what becomes of the distributions of $\left[\vb{G}_{\vb{M}}^{N}\right]_{11}$ or $\g_{\vb{M}}^{N}$ when $\vb{M}$ is no longer Hermitian. Do we still observe this inversion symmetry? What can we say about the tails?

At this stage, we are unable to compute explicitly the limiting distribution $\g_{\vb{M}}^N$ for general non-Hermitian matrices or even for Ginibre matrices to find the analogue of the Cauchy distribution in Eq. \eqref{dis} for Wigner matrices. We have only conjectures, which we will detail in Section \ref{conjectures} below.

However, we succeeded in calculating the distribution of a diagonal element of the resolvent, namely $\left[\vb{G}_{\vb{M}}^{N}(z)\right]_{11}$, in the case when $\vb{M}$ is a complex Ginibre matrix. We recall that The Lebesgue density of a $N\times N$ complex Ginibre matrix called $\vb{M}$ in $\mathcal{M}_n(\C) = \C^{N\times N}$ is 
\begin{equation}
    \vb{A}\mapsto \frac{1}{\pi^{N^2}}\exp{-N\Tr(\vb{A}^*\vb{A})} = \frac{1}{\pi^{N^2}}\exp{-N\sum_{i, j=1}^N |\vb{A}_{i, j}|^2} 
\end{equation}
where $\vb{A}^*$ denotes the conjugate-transpose of $\vb{A}$.

The paper is organized as follows. First, we present our main result: the theoretical distribution of $\left[\vb{G}_{\vb{M}}^{N}(z)\right]_{11}$ for finite $N$ when $\vb{M}$ is a Ginibre matrix. We then let $N$ tend to infinity, examining cases where $z$ lies inside the unit disk, outside it, or approaches the unit circle under various scaling regimes in $N$, accompanied by numerical results. Next, we show that, in the high-dimensional limit, the tail of the diagonal elements of the resolvent no longer reflects the spectral distribution but instead depends on the statistics of the left and right eigenvectors through the so-called self-overlaps. We then present two conjectures for the case where $\vb{M}$ is no longer Ginibre: one concerning the limiting distribution of the diagonal elements of the resolvent, and the other concerning the distribution of  $\g_{\vb{M}}^N(z)$. The remainder of the paper is dedicated to the proof of our main result.

\section{The distribution of the resolvent diagonal elements}

The following theorem is our main result, which provides the explicit calculation of the distribution of $\left[\vb{G}_{\vb{M}}^{N}\right]_{11}$ for finite $N$.
\begin{theorem}\label{theorem1}
    Let $\vb{M}$ be a $N\times N$ complex Ginibre matrix with $N\geq 2$. We define 
    \begin{equation}
        [\vb{G}_{\vb{M}}^N(z)]_{11} := [(z\vb{1} - \vb{M})^{-1}]_{1 1}
    \end{equation}for $z$ outside the spectrum of $\vb{M}$.
    Then, we have
    \begin{equation}\label{not}
        \frac{1}{[\vb{G}_{\vb{M}}^N(z)]_{11}} - z \sim \mathcal{N}_{\C}\left(0, \frac{1+t}{N}\right)
    \end{equation}
    where $t$ admits a probability density function given by

    \begin{multline}\label{sigma}
    t\in \R^{+}\mapsto \frac{e^{-\frac{N|z|^2}{\left(1 + \frac{1}{t}\right)}}}{(N-2)!} \frac{t^{N-2}}{(1 + t)^N}
    \left(\Gamma(N-1, N|z|^2)e^{N|z|^2}\left(N-1 - \frac{tN}{1+t}|z|^2\right)  +\left(N|z|^2\right)^{N-1}\right)
\end{multline}

where $\Gamma$ is the incomplete Gamma function defined as 
\begin{equation}
    \Gamma(s, x) = \int_x^{+\infty} u^{s-1}e^{-u}du.
\end{equation}
\end{theorem}

The proof relies on the cavity method, which enables us to express $1/[\vb{G}_{\vb{M}}^N(z)]_{11} - z$ as a mixture of Gaussians, where the distribution of its variance is dictated by $t$. While the distribution of $t$ in the theorem may seem complex at first sight, it can be derived through a straightforward computation using characteristic functions.

The objective of the next theorem is to analyze the limit of \eqref{sigma} as $N \to \infty$ under different regimes of $z$. It turns out that when $|z| < 1$, the distribution of $t$ converges to an inverse gamma law. Consequently, in the limit $N \to \infty$, the distribution of $[\vb{G}_{\vb{M}}^N(z)]_{11}$ becomes a complex Student variable. This result constitutes the first point of the following theorem.

The detailed computational steps, which provide the proofs of the preceding and subsequent results, are included in the appendix.

\begin{theorem}\label{main}
    Let $\vb{M}$ be a $N\times N$ complex Ginibre matrix. We define 
    \begin{equation}
        [\vb{G}_{\vb{M}}^N(z)]_{11} := [(z\vb{1} - \vb{M})^{-1}]_{1 1}
    \end{equation}
    for $z$ outside the spectrum of $\vb{M}$. We then have several limiting results for different regimes.
    \begin{regimes}
        \item If $|z|<1$ fixed, then the distribution of $[\vb{G}_{\vb{M}}^N(z)]_{11}$   
    converges in law when $N\to \infty$ towards a complex Student law with $\nu = 2$ with a density function that is
    \begin{equation}\label{student}
        \omega\in \C\mapsto \frac{1}{\pi}\frac{1 - |z|^2}{(1 - |z|^2 + |\omega - \overline{z}|)^2}.
    \end{equation}
     More generally, if $|z|^2 = 1 - g(N)$ with $g(N)\geq 0$ such that $\frac{1}{N}\ll g(N)$, then 
    \begin{equation}\label{regime2}
        \frac{1}{\sqrt{g(N)}}\left([\vb{G}_{\vb{M}}^N(z)]_{11} - \overline{z}\right) \xrightarrow{\mathcal{L}} \Omega
    \end{equation}where $\Omega$ is a complex Student variable with $\nu=2$ whose density function equal to 
    \begin{equation}\label{stu}
        \omega\in \C\mapsto \frac{1}{\pi}\frac{1}{(1 + |\omega|^2)^2}.
    \end{equation}
     \item If $|z|^2 = 1 + \epsilon(N)$ with $\epsilon(N) = o\left(\frac{1}{\sqrt{N}}\right)$, then 
    \begin{equation}\label{regime3}
        N^{1/4}\left([\vb{G}_{\vb{M}}^N(z)]_{11} -1\right) \xrightarrow{\mathcal{L}} \mathcal{N}_{\C}(0, t)
    \end{equation}where $t$ admits a probability density function given by
    \begin{equation}
        t\in \R^{+}\mapsto \frac{1}{t^2}e^{-\frac{1}{2t^2}}\left(\frac{1}{2t} + \frac{1}{\sqrt{2\pi}}\right).
    \end{equation}
    \item If $|z|^2 = 1 + \frac{\alpha}{\sqrt{N}}$ with $\alpha\in \R$, then 
    \begin{equation}\label{regime4}
        N^{1/4}\left([\vb{G}_{\vb{M}}^N(z)]_{11} - 1\right) \xrightarrow{\mathcal{L}} \mathcal{N}_{\C}(0, t)
    \end{equation}where $t$ admits a probability density function given by
    \begin{equation}
        t\in \R^{+}\mapsto \frac{1}{t^2}e^{-\frac{1}{2t^2} + \frac{\alpha}{t}}\left(\frac{1 - \alpha t}{t \sqrt{2\pi}}\int_{\alpha}^{+\infty}e^{-v^2/2}dv  + \frac{e^{-\frac{\alpha^2}{2}}}{\sqrt{2\pi}}\right).
    \end{equation}
    \item If $|z|^2 = 1 + f(N)$ with $\frac{1}{\sqrt{N}}\ll f(N)$, then 
    \begin{equation}\label{regime5}
        \sqrt{Nf(N)|z|^2}\left([\vb{G}_{\vb{M}}^N(z)]_{11} - \frac{1}{z}\right) \xrightarrow{\mathcal{L}} \mathcal{N}_{\C}(0,1).
    \end{equation}
    In particular, if $|z|>1$ fixed, then 
    \begin{equation}\label{regime6}        \sqrt{N}\left([\vb{G}_{\vb{M}}^N(z)]_{11} - \frac{1}{z}\right) \xrightarrow{\mathcal{L}} \mathcal{N}_{\C}\left(0, \frac{1}{|z|^2(|z|^2 - 1)}\right).
    \end{equation}
    \end{regimes}
\end{theorem}

\paragraph{Remarks}

\begin{itemize}
    \item As the statistics of $\vb{M}$ only depends on $\Tr(\vb{M}^*\vb{M})$, the law of each diagonal element follows the same law as $[\vb{G}_{\vb{M}}(z)]_{11}$. Furthermore, the bi-invariance of $\vb{M}$ \footnote{In the sense that $ \vb{M} $ and $ \vb{U} \vb{M} \vb{V} $ follow the same distribution for all unitary matrices $ \vb{U} $ and $ \vb{V} $.}
 implies that all these distributions depend only on $|z|$. 
    \item The expected value of $ [\vb{G}_{\vb{M}}^N(z)]_{11} $ is $ \mathbb{E}(\g_{\vb{M}}^N(z)) $, which converges to $ \g_{M}(z) $. So we expect all the limiting distributions described in the previous theorem to have expectations of $\overline{z}$ when $z$ is within the unit disk and $\frac{1}{z}$ otherwise, see Eq. \eqref{eq:ginibre}. One can check that this is indeed the case. 

    \item The presence of eigenvalues around the point $z$ leads to the dispersion of the random variable $[\vb{G}_{\vb{M}}^N(z)]_{11}$. The worst-case scenario occurs when $z$ is inside the unit disk since $[\vb{G}_{\vb{M}}^N(z)]_{11}$ does not even have a finite variance. As we move away from the bulk, $[\vb{G}_{\vb{M}}^N(z)]_{11}$ concentrates around its mean.
    
    \item Within the bulk of the spectrum $ |z| < 1 $, the limiting distribution of $ [\vb{G}_{\vb{M}}(z)]_{11} $ is a generalization of the Cauchy distribution (Eq. \eqref{dis}). Interestingly, $ [\vb{G}_{\vb{M}}(z)]_{11} $ and $ 1/\overline{[\vb{G}_{\vb{M}}(z)]_{11}} $ have the same distribution by Lemma \ref{inverse}. Furthermore, we can more easily show that $\Omega$ and $ 1/\Omega
    $ follow the same distribution. 

    \item We can recover regime (2) by taking $\alpha = 0$ in regime (3).
    
    \item The result \eqref{regime2} is an interesting theoretical finding demonstrating the robustness of the Student fluctuation (Eq. \eqref{student}) when $ z $ is close to the unit circle but inside the disk. Similarly, the result \eqref{regime5} shows the robustness of the Gaussian behavior (Eq. \eqref{regime6}) when $ z $ is close to the unit circle but outside the disk.

    \item The tails arise from the presence of an eigenvalue very close to $z$. We will analyze in detail the tails of regime (1) in the next section. Note that such tails continue to exist up to the regime $|z|^2 = 1 + \frac{\alpha}{\sqrt{N}}$. This is because when $N$ is finite, there may well be eigenvalues slightly larger than $1$. More quantitatively, one can read in \cite{rider2003limit} that the largest absolute value of the eigenvalues of an $N\times N$ complex Ginibre matrix can be approximated by

\begin{equation}
    1 + \sqrt{\frac{\gamma_N}{N}} - \frac{1}{\sqrt{4N \gamma_N}} \log(Z)
\end{equation}

where $Z$ is an exponential random variable with mean 1, $\gamma_N = \log \frac{N}{2\pi} - 2 \log \log N$, and $N \to \infty$.
\end{itemize}

\section{Numerical results of the resolvent diagonal elements}

We numerically verify the results of Theorem \ref{main} for the four regimes.

\begin{figure}[H]
    \centering
    \begin{subfigure}{0.47\textwidth}
        \centering
        \includegraphics[width=\textwidth]{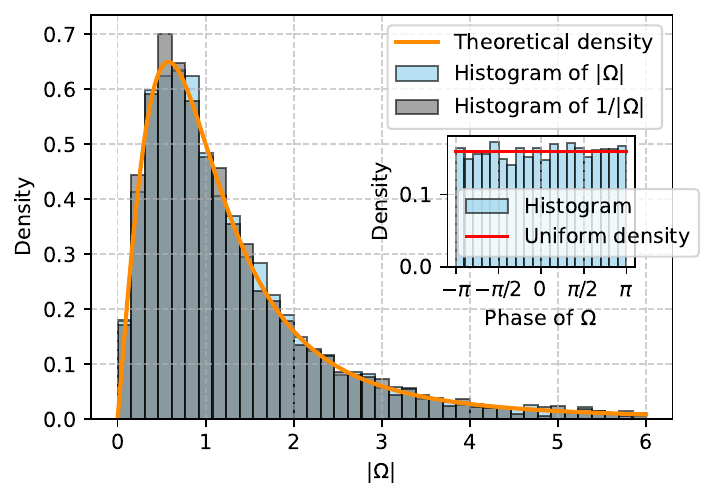}
        \caption{Distribution of $\Omega$.}
    \end{subfigure}
    \hfill
    \begin{subfigure}{0.47\textwidth}
        \centering
        \includegraphics[width=\textwidth]{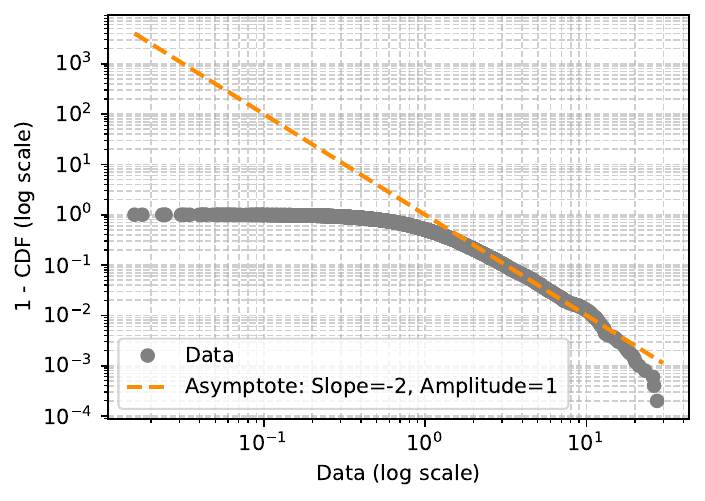}
        \caption{Log-log plot of the tail of $\Omega$'s distribution.}
    \end{subfigure}
    \caption{Model (1). Simulation of $5000$ complex Ginibre matrices of size $1000$ for the variable $\Omega:= \frac{1}{\sqrt{1 - |z|^2}}\left([\mathbf{G}_{\mathbf{M}}(z)]_{11} - \overline{z}\right) $ with $ z = 0.7 $. The distribution of $\Omega$ is compared with $1/\Omega$ (which is expected to be the same according to the third remark of the previous theorem), as well as with the theoretical density \eqref{stu}, whose tail is given by $ 1/G^2 $.}

\end{figure}

\begin{figure}[H]
    \centering
    \begin{subfigure}{0.47\textwidth}
        \centering
        \includegraphics[width=\textwidth]{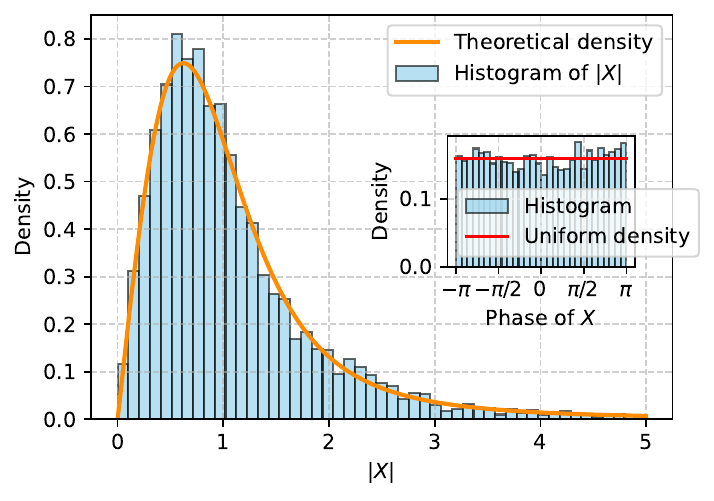}
        \caption{Distribution of $ X $.}
    \end{subfigure}
    \hfill
    \begin{subfigure}{0.47\textwidth}
        \centering
        \includegraphics[width=\textwidth]{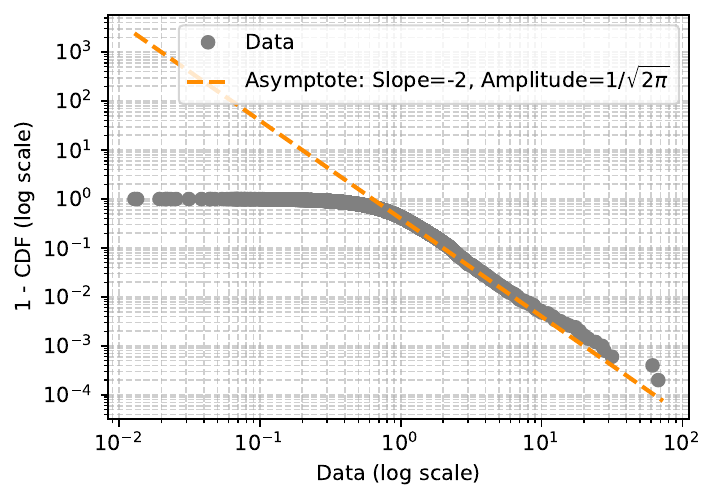}
        \caption{Log-log plot of the tail of $ X $'s distribution.}
    \end{subfigure}
    \caption{Model (2). We simulate $5000$ complex Ginibre matrices of size $1000$ and study the variable $X:=N^{1/4}\left([\mathbf{G}_{\mathbf{M}}(z)]_{11} - 1/z\right)$ where we have chosen $z = \sqrt{1 + 1/N}$. $X$ should follow the theoretical density \eqref{regime3} which is plotted using numerical integration. The theoretical tail is $1/(\sqrt{2\pi}G^2)$.}
\end{figure}

\begin{figure}[H]
    \centering
    \begin{subfigure}{0.47\textwidth}
        \centering
        \includegraphics[width=\textwidth]{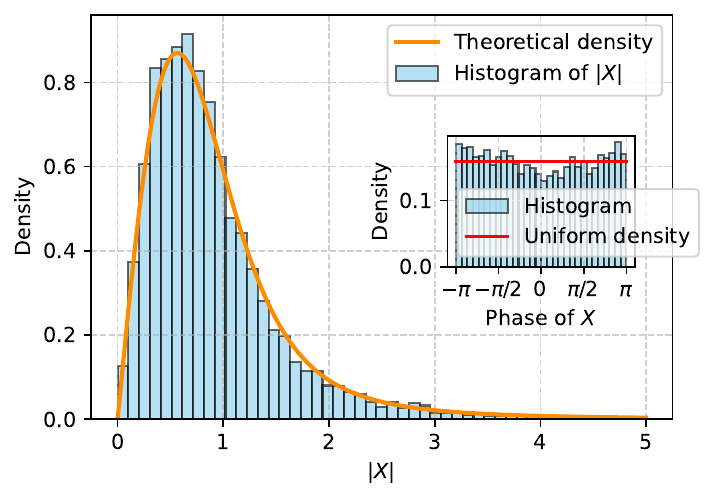}
        \caption{Distribution of $ X $.}
    \end{subfigure}
    \hfill
    \begin{subfigure}{0.47\textwidth}
        \centering
        \includegraphics[width=\textwidth]{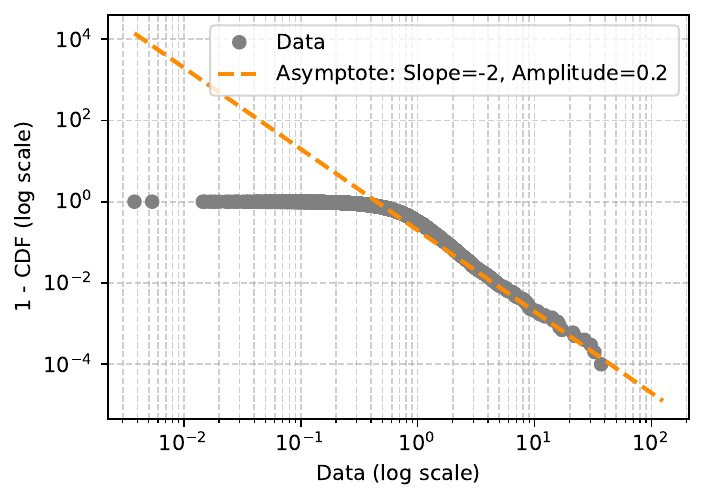}
        \caption{Log-log plot of the tail of $ X $'s distribution.}
    \end{subfigure}
    \caption{Model (3). We simulate $5000$ complex Ginibre matrices of size $1000$ and study the variable $ X := N^{1/4}\left([\mathbf{G}_{\mathbf{M}}(z)]_{11} - \frac{1}{z}\right) $, where we have chosen $ z = \sqrt{1 + \frac{1}{2\sqrt{N}}} $. The variable $ X $ is expected to follow the theoretical density given in \eqref{regime4}, which is plotted using numerical integration. The theoretical tail is $ \frac{1}{G^2}\left(\frac{\alpha}{2} \operatorname{erfc}\left(\frac{\alpha}{\sqrt{2}}\right) + \frac{\exp\left(-\frac{\alpha^2}{2}\right)}{\sqrt{2\pi}}\right) $ (where $ \operatorname{erfc} $ is the complementary error function), which is approximated as $ 0.2 $ when $ \alpha = \frac{1}{2} $.}

\end{figure}

\begin{figure}[H]
    \centering
    \begin{subfigure}{0.47\textwidth}
        \centering
        \includegraphics[width=\textwidth]{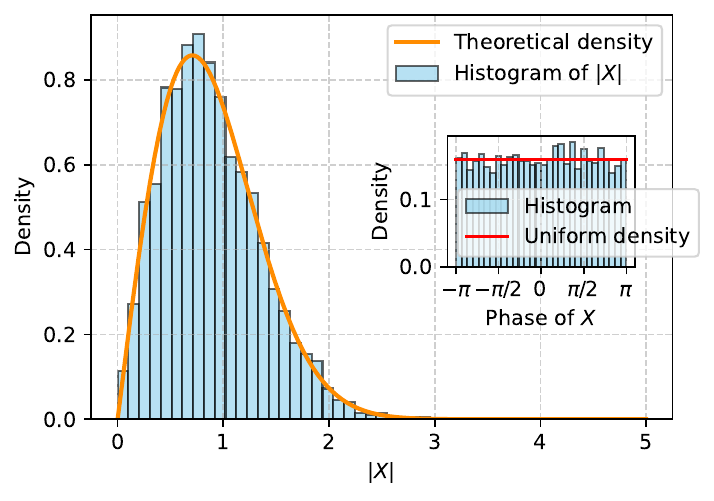}
        \caption{}
    \end{subfigure}
    \caption{Model (4). We simulate $10000$ complex Ginibre matrices of size $1000$ and study the variable $X:=\sqrt{N} \sqrt{|z|^2(1 - |z|^2)}\left([\mathbf{G}_{\mathbf{M}}(z)]_{11} - 1/z\right)$ where we have chosen $ z = 2 $. $ X $ should follow the theoretical density given by a standard complex Gaussian distribution according to Regime (4) of Theorem \ref{main}.}
\end{figure}

\section{Analysis of tails in the bulk region}

\subsection{Tail of the Stieltjes transform}

When $\vb{M}$ is a Hermitian matrix, the ideas from \cite{cizeau1994theory} recalled in the introduction yield

\begin{equation}
    P(|\g_{\vb{M}}(x)| \geq G) \underset{G \to +\infty}{\sim} \frac{\rho_{\vb{M}}(x)}{G},
\end{equation}
where $x$ is inside the support of the spectrum.

We now extend the argument to the case where $\vb{M}$ is a non-Hermitian matrix and $z \in \mathbb{C}$ lies in the bulk of the complex spectrum. When the eigenvalues of $\vb{M}$ are complex, the situation is quite different from the Hermitian case because $\g_{\vb{M}}$ can now be defined inside the spectrum. When $N$ is large, $z$ will always be close to an eigenvalue $\lambda$ of $\vb{M}$. The Stieltjes transform is thus dominated by that particular $\lambda$:

\begin{equation}
    |\g_{\vb{M}}^{N}(z)| \approx \frac{1}{N} \frac{1}{|z - \lambda|}.
\end{equation}
One can thus write

\begin{align}\label{firsttail}
    \mathbb{P}\left(|\g_{\vb{M}}^{N}(z)| \geq \frac{G}{\sqrt{N}}\right) = \mathbb{P}\left( |z - \lambda| \leq \frac{1}{\sqrt{N}G}\right) \approx N \rho_{\vb{M}}(z) \frac{\pi}{N G^2}.
\end{align}

This suggests that, again, the tail of the distribution of the Stieltjes transform contains information about the local density of eigenvalues, but now decays as $\pi \rho_{\vb{M}}(z)/G^3$. This tail is embedded within the fluctuations of $\g_{\vb{M}}^{N}(z)$, which, unlike in the Hermitian case, converges to a limit $\g_{\vb{M}}$. In the non-Hermitian case, it appears more appropriate to study the random variable

\begin{equation}\label{defs}
    \s_{\vb{M}}^N(z) := \sqrt{N}(\g_{\vb{M}}^N(z) - \g(z))
\end{equation}
and particularly its limit, denoted by $\s_{\vb{M}}$, as $N \to +\infty$. The distribution of $\s_{\vb{M}}(z)$ will be the subject of a conjecture in the next section. Thanks to Equation \eqref{firsttail}, we have

\begin{equation}\label{tailcomp}
    P(|\s_{\vb{M}}(z)| \geq r) \underset{r \to +\infty}{\sim} \frac{\pi \rho_{\vb{M}}(z)}{r^2}.
\end{equation}

Note that this tail is different from that obtained for $ |[\vb{G}_{\vb{M}}(z)]_{11}| $. Indeed, Theorem \ref{main} gives

\begin{equation}
\mathbb{P}(|[\vb{G}_{\vb{M}}(z)]_{11}| \geq G) \underset{G\to +\infty}{\approx} \frac{1 - |z|^2}{G^2}.
\end{equation}

In order to understand the extra factor $1 - |z|^2$, we need to introduce the concept of eigenvector overlaps.

\subsection{Left and right eigenvector overlaps}

In the case of Ginibre matrices, $\vb{M}$ is in general not normal (i.e. $\vb{M} \vb{M}^* \neq \vb{M}^* \vb{M}$). For a random matrix $\vb{M}$, we can safely assume that all its $N$ eigenvalues have multiplicity one, ensuring $\vb{M}$ is still diagonalisable by a transformation to its eigenbasis $\vb{M} = \vb{P} \vb{D} \vb{P}^{-1}$\label{diagon} with $\vb{D} = \diag(z_1, ..., z_n)$ its complex eigenvalues. We know that in general $\vb{P}^{-1}\neq \vb{P}^*$. 

We define the right eigenvectors $\ket{R_i}$ as the columns of the matrix $\vb{P}$ and the left eigenvectors $\bra{L_i}$ and as the rows of $\vb{P}^{-1}$. Thus, we can rewrite \eqref{diagon} as 
\begin{equation}
    \vb{M} = \sum_{i=1}^N \lambda_i \ket{R_i}\bra{L_i}.
\end{equation}
Since, $\vb{P} \vb{P}^{-1} = I_N$, we get
\begin{equation}
    \bra{R_i} \ket{L_j} = \delta_{i, j}.
\end{equation}
When the matrix $\vb{M}$ is normal, we further have $\ket{R_i} = \ket{L_i}$. 

The first object that allows qualitative insights into the non-orthogonality of the eigenvectors, is the so-called overlap matrix $\mathcal{O}$ which is defined as \cite{chalker1998eigenvector}
\begin{equation}
    \mathcal{O}_{i, j} = \bra{L_j}\ket{L_i} \bra{R_j}\ket{R_i}.
\end{equation}
The diagonal entries of $\mathcal{O}$ are called self-overlaps. The exploration of these overlaps is motivated by multiple applications, especially in physics. For instance, these self-overlaps are useful to understand the sensitivity of the eigenvalues of $\vb{M}$ to a perturbation. Note that for Hermitian matrices, $ \mathcal{O}_{i, j} \equiv \delta_{i,j}$.

One way to study the statistical properties of these quantities is to look at the single-point function \cite{chalker1998eigenvector}
\begin{equation}\label{for}
    \mathcal{O}(z) = \left<\frac{1}{N}\sum_n \mathcal{O}_{nn} \delta(z - z_n)\right>.
\end{equation}
where the average is taken with respect to the probability of the random matrix $\vb{M}$ and $\delta$ is the complex delta function.

In \cite{chalker1998eigenvector}, Chalker and Mehlig computed the leading asymptotic behaviour of $\mathcal{O}(z)$ as $N$ goes to infinity when $\vb{M}$ is a complex Ginibre ensemble, which reads
\begin{equation}\label{leading}
    \mathcal{O}(z) \approx \frac{N (1 - |z|^2)}{\pi}
\end{equation}inside the unit disk. Note that the overlap is of order $N$ and not of order 1 like for Hermitian matrices.
For more general non-normal matrices, one can think that $\mathcal{O}(z)$ will be of order $N$, thus we introduce the rescaled self-overlap $\Tilde{\mathcal{O}}(z) = N^{-1}{\mathcal{O}(z)}$, which is expected to converge to a limit as $N$ tends to infinity.

Moreover, ${\mathcal{O}}(z)$ being the result of a local average, one can we study $\Tilde{\mathcal{O}}_{nn}$ as a random variable. In particular,  Bourgade and  Dubach \cite{bourgade2020distribution} proved that conditionally on $z_1 = z$, the rescaled diagonal overlap $\Tilde{\mathcal{O}}_{11}$ converges in distribution to an inverse Gamma random variable as $N$ goes to infinity when $\vb{A}$ is a complex Ginibre matrix. More precisely, they proved that
\begin{equation}\label{inverse_gamma}
    \frac{\mathcal{O}_{11}}{N(1- |z|^2)} \underset{N\to \infty}{\longrightarrow} X
\end{equation}
where $X$ is a random variable with probability density
\begin{equation}\label{secontail}
    x\mapsto\frac{1}{x^{3}} e^{-\frac{1}{x}}.
\end{equation}
The heavy tail $\frac{1}{x^3}$ was predicted in \cite{mehlig2000statistical} and the first moment of this distribution is the precisely leading order term of $\mathcal{O}(z)$ given in Eq. \eqref{leading}.

In 2018,  Fyodorov \cite{fyodorov2018statistics} developed a method to compute and analyze these overlaps, investigating bulk and edge scaling limits and recovered this result. Additionally,  Fyodorov used the same method to understand the overlaps for other ensembles: the real Ginibre ensemble \cite{wurfel2310mean} and the complex elliptic Ginibre ensemble \cite{crumpton2024mean}.

\subsection{A tale of two tails?}

A priori, the tail of $ G = [\vb{G}_{\vb{M}}^N(z)]_{11} $ could arise from the combination of two rare events: an eigenvalue $ z_n $ very close to $ z $ and an exceptionally large value for the associated overlap $ \mathcal{O}_{nn} $. Therefore, we need to understand how the tail of Eq. \eqref{firsttail} and that of Eq. \eqref{secontail} combine. A general, direct argument, however, suggests that provided the mean value of the distribution of $ \Tilde{\mathcal{O}}_{nn} $ is finite and equal to $ \overline{O} $, the tail of the distribution of $ G $ is given by $ \overline{O}/G^2 $ and is dominated by the proximity of eigenvalue $ z_n $ and not by the tail of Eq. \eqref{secontail}.

Let us make this argument more precise. Using left and right eigenvectors, one can write
\begin{equation}\label{sum}
   [\vb{G}_{\vb{M}}(z)]_{11} = \sum_{i=1}^N \frac{1}{z - \lambda_i} [L_i]_1 [R_i]_1.
\end{equation}
When $ N $ is large, $ z $ will always be close to an eigenvalue $ z_n $ of $ \vb{M} $ as $ N $ goes to infinity. Thus, the sum \eqref{sum} is dominated by this particular term:
\begin{equation}\label{approximation}
    [\vb{G}_{\vb{M}}(z)]_{11} \approx \frac{[L_n]_1 [R_n]_1}{z - z_n} 
\end{equation}
or in other words
\begin{equation}
    \left\lvert[\vb{G}_{\vb{M}}(z)]_{11}\right\rvert^2 \approx  \left\lvert\frac{[L_n]_1 [R_n]_1}{z - z_n} \right\rvert^2.
\end{equation}

Using the rotational invariance of $ \vb{M} $\footnote{We say that a matrix $ \vb{M} $ is rotationally invariant if $ \vb{M} $ and $ \vb{U} \vb{M} \vb{U}^* $ follow the same distribution for all unitary matrices $ \vb{U} $.} and the fact that there is only one constraint on $ L_n, R_n $, which is that $ \bra{L_n}\ket{R_n} = 1 $, we can consider that $ |[L_n]_1 [R_n]_1|^2 $ will behave like $ \mathcal{O}_{nn}/N^2 $ conditioned to $ z = z_n $. Using Eq. \ref{inverse_gamma}, we thus know that given $ z_n $, the probability density function of the variable
\begin{equation}
    U = \left\rvert\frac{[L_n]_1 [R_n]_1}{z - z_n} \right\rvert^2
\end{equation}
is 
\begin{equation}
    u\mapsto\frac{(1 - |z|^2)^2}{N^2 |z - z_n|^4 u^3} \exp\left(-\frac{(1 - |z|^2)}{N |z - z_n|^2 u}\right).
\end{equation}
Furthermore, a slight generalization of Eq. \eqref{firsttail} yields 
\begin{equation}
    P(\sqrt{N}|z - z_n|\leq a) \underset{N\to \infty}{\longrightarrow} 1 - \exp(-\pi \rho(z) a^2).
\end{equation}
In other words, $ {N}|z - z_n|^2 $ converges in distribution to a random variable with a probability density function given by
\begin{equation}
    x\mapsto \pi \rho(z) \, \exp(-\pi \rho(z) \, x).
\end{equation}
Hence, the unconditional limiting density function of $ U $ is given by
\begin{equation}
    f := u\mapsto\int_0^{+\infty} \frac{(1 - |z|^2)^2 \pi \rho(z)}{x^2 u^3} \exp\left(-\frac{(1 - |z|^2)}{u x}\right) e^{-\pi \rho(z) x} \, {\rm d}x.
\end{equation}
The change of variables $ v = ux $ then gives
\begin{equation}
    f(u) =  \frac{(1 - |z|^2)^2 \pi \rho(z)}{u^2} \int_0^{+\infty} \frac{1}{v^2} \exp\left(-\frac{(1 - |z|^2)}{v}\right) e^{-\frac{\pi \rho(z)}{u v}} \, {\rm d}v.
\end{equation}
Since 
\begin{equation}
    \int_0^{+\infty} \frac{1}{v^2} \exp\left(-\frac{(1 - |z|^2)}{v}\right) e^{-\frac{\pi \rho(z)}{u v}} \, {\rm d}v \underset{u\to \infty}{\longrightarrow} \int_0^{+\infty}  \frac{1}{v^2} \exp\left(-\frac{(1 - |z|^2)}{v}\right) {\rm d}v = \frac{1}{1 - |z|^2},
\end{equation}
we finally have
\begin{equation}\label{tail}
    f(u)\underset{u\to \infty}{\sim} \frac{(1 - |z|^2) \pi \rho(z)}{u^2} = \frac{1 - |z|^2}{u^2}
\end{equation}
since $ \pi \rho(z) = 1 $ when $ \vb{M} $ is a Ginibre matrix.

Returning to our approximation \eqref{approximation} using Eq. \eqref{tail}, we can then conclude that our non-rigorous reasoning in this section allows us to predict the correct tail behavior of $ [\vb{G}_{\vb{M}}^N(z)]_{11} $ predicted by the theorem, which is
\begin{equation}
    P(|[\vb{G}_{\vb{M}}^N(z)]_{11}|\geq G) \underset{G\to \infty}{\sim} \frac{1 - |z|^2}{G^2},
\end{equation}
precisely as the result we surmised using a general argument about the tail of variables of the form $ \sqrt{\mathcal{O}(z)}/(z - z_n) $.

\section{Conjectures}\label{conjectures}

\subsection{On the resolvent diagonal elements}

In fact, the previous reasoning prompts us to conjecture a certain universality, namely a generalization of Eq. \eqref{student}:
\begin{conjecture}\label{first_conjecture}
    Let $\vb{M}$ be a rotationally invariant random matrix of size $N$ which admits a non-pathological limiting spectral density. We define 
    \begin{equation}
        [\vb{G}_{\vb{M}}^N(z)]_{11} := [(z\vb{1} - \vb{M})^{-1}]_{1 1}
    \end{equation}
    for $z$ outside the spectrum of $\vb{M}$. Then,the distribution of $[\vb{G}_{\vb{M}}^N(z)]_{11}$ converges in law when $N\to \infty$ towards a complex Student law with $\nu = 2$ with a density function given by
    \begin{equation}\label{student2}
        \omega\in \C\mapsto \frac{1}{\pi}\frac{\beta}{(\beta + |\omega - \g_{\vb{M}}(z)|)^2}.
    \end{equation}
    where $\beta = \pi\Tilde{\mathcal{O}}(z)$ where $\Tilde{\mathcal{O}}(z)$ is the rescaled self-overlap of $\vb{M}$ corresponding to $z$ and $\g_{\vb{M}}$ the limiting Stieltjes transform of $\vb{M}$.

   In other words, the distribution of

\begin{equation}
    \Omega := \frac{[\vb{G}_{\vb{M}}^N(z)]_{11} - \g_{\vb{M}}(z)}{\sqrt{\pi \Tilde{O}(z)}}
\end{equation}
no longer depends on $\vb{M}$ or $z$, and its density function is given by

\begin{equation}\label{t}
    \omega \in \mathbb{C} \mapsto \frac{1}{\pi} \frac{1}{(1 + |\omega|^2)^2}.
\end{equation}

\end{conjecture}

\paragraph{Remarks}

\begin{itemize}
    \item When $\vb{M}$ is a Ginibre matrix, we recover the distribution \eqref{student} since $\Tilde{\mathcal{O}}(z) = \frac{1 - |z|^2}{\pi}$ using the results recalled in \eqref{leading}.
    
    \item The tail of $[\vb{G}_{\vb{M}}^N(z)]_{11}$ is thus given by
    \begin{equation}
    P(|[\vb{G}_{\vb{M}}^N(z)]_{11}| \geq G) \underset{G\to \infty}{\sim} \frac{\pi\Tilde{O}(z)}{G^2}.
    \end{equation}

    \item Note that $Y$ and $1/Y$ follow the same distribution, providing a very efficient method to numerically verify our conjecture.

    \item Figure \eqref{veri} aims to verify this conjecture for a non-trivial example using a GinUE matrix\footnote{Recall that the complex elliptic Ginibre ensemble (GinUE) consists of mean-zero i.i.d. complex Gaussian entries with the following correlation structure: $ E(X_{ii}^2) = \frac{1}{N} $ and $ E(X_{ij} X_{ji}) = \frac{\tau}{N} $. }. So, we need to estimate $\beta$. It is important to note that if we average the rescaled self-overlaps associated with the complex eigenvalues closest to $ z $, within a radius of $ \frac{1}{N^{0.25}} $ from a large sample set of $ M $, then we will be estimating

\begin{equation}
    \mathbb{E}(\Tilde{\mathcal{O}}_{nn} \mid z = z_n) = \frac{\tilde{\mathcal{O}}(z)}{\rho(z)},
\end{equation}
which is not equal to $ \tilde{\mathcal{O}}(z) $. Therefore, if we wish to estimate $ \beta $ directly, we need to multiply the previous estimate by $ \pi \rho(z) $, which can be readily approximated.

\end{itemize}

\begin{figure}[H]
    \centering
    \begin{subfigure}{0.47\textwidth}
        \centering
        \includegraphics[width=\textwidth]{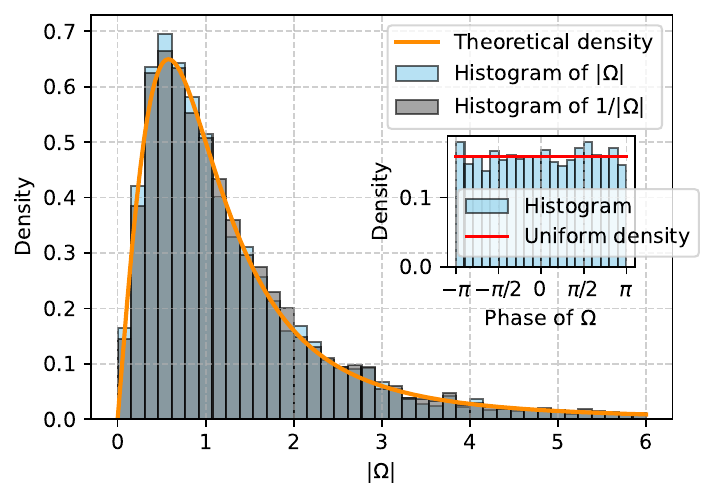}
        \caption{Distribution of $\Omega$.}
    \end{subfigure}
    \hfill
    \begin{subfigure}{0.47\textwidth}
        \centering
        \includegraphics[width=\textwidth]{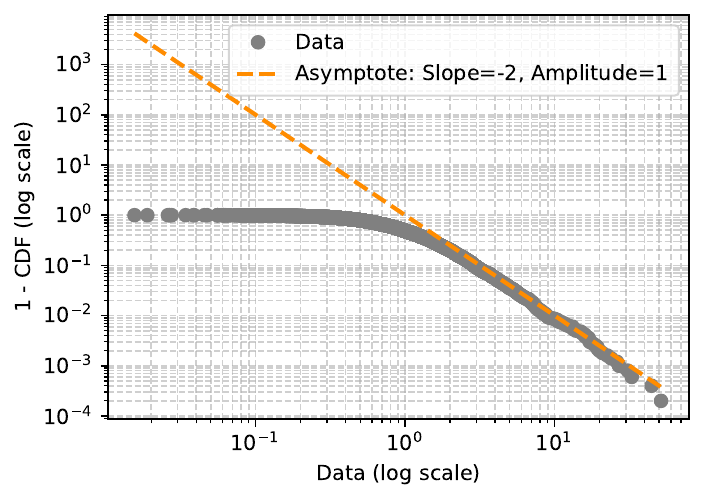}
        \caption{Log-log plot of the tail of $\Omega$'s diribution.}
    \end{subfigure}
    
    \caption{We choose a non-trivial example by taking a matrix $ \vb{M} = \vb{A}\vb{B}\vb{C} $ of size $ 300 $, where $ \vb{A} $ is a GinUE matrix with parameter $ \tau = 0.5 $, $ \vb{B} $ is a unitary matrix randomly chosen according to the Haar measure, and $ \vb{C} $ is a complex Ginibre matrix. We consider a set of $ 5000 $ samples of $ \vb{M} $. Let $ z = 0.4i + 0.2 $. We study $\Omega:= ({[\mathbf{G}_{\mathbf{M}}(z)]_{11} - \alpha})/{\sqrt{\beta}} $, where $ \alpha $ is the mean of $ [\mathbf{G}_{\mathbf{M}}(z)]_{11} $ and $ \beta \approx 1.28$ is obtained by the method described in the previous remark using $ 3000 $ samples of $ \vb{M} $. We expect $\Omega$ to follow the theoretical density given in \eqref{t}, whose tail is given by ${1}/{G^2} $, and which is also expected to match the distribution of $ {1}/{\Omega} $.
 }    
 \label{veri}
\end{figure}

\subsection{On the Stieltjes transform}

At this stage, we are unable to formulate a conjecture on the distribution of $\s_{\vb{M}}$ introduced in Eq. \eqref{defs}, even if $\vb{M}$ is a Ginibre matrix. In this case, we showed in Eq. \eqref{student} that each diagonal entry of the resolvent of $\vb{M}$ follows a Student distribution with a tail given by $\frac{\Tilde{O}(z)}{G^2}$. Note that $\s_{\vb{M}}^N(z)$ is simply

\begin{equation}
    \frac{1}{\sqrt{N}} \sum_{k=1}^N ([\vb{G}_{\vb{M}}^N(z)]_{kk} - \overline{z}).
\end{equation}
Thus, we are summing $N$ Student distributions with a tail of $\frac{\pi\Tilde{O}(z)}{G^2} = \frac{1 - |z|^2}{G^2}$, which should yield a distribution with a tail of

\begin{equation}
    \frac{\rho_{\vb{M}}(z)\pi}{G^2} = \frac{1}{G^2}
\end{equation}
according to Eq. \eqref{tailcomp}. This necessarily implies that the Student distributions must be positively correlated in such a way that the amplitude of the tail passes from $ 1 - |z|^2 $ to $ 1 $. The reasoning can be adapted to the general case by Conjecture \ref{first_conjecture}.

In fact, based on numerical observations, we propose the following conjecture on $ \s_{M}^N(z) $, or rather $ \sqrt{\pi \rho_{\vb{M}}(z)} \s_{M}^N(z) $.

\begin{conjecture}
    Let $\vb{M}$ be a rotationally invariant random matrix of size $N$ with a well-behaved limiting spectral density. We define

    \begin{equation}\label{defZ}
        Z := \lim_{N \to +\infty} \sqrt{N} \frac{\left(\g^N_{\vb{M}}(z) - \g_{\vb{M}}(z)\right)}{\sqrt{\pi \rho_{\vb{M}}(z)}},
    \end{equation}
where $z$ lies within the spectrum, $\g_{\vb{M}}$ is the limiting Stieltjes transform of $\vb{M}$, and $\rho_{\vb{M}}$ represents the limiting spectral distribution. Then, the distribution of $Z$ is independent of $\vb{M}$ and $z$, and satisfies
    \begin{equation}\label{que}
        P(|Z| \geq r) \underset{r \to +\infty}{\sim} \frac{1}{r^2},
    \end{equation}
where $Z$ and ${\sqrt{2}}/{Z}$ follow the same distribution.
\end{conjecture}
\paragraph{Remarks}
\begin{enumerate}
    \item Note that $Z$ cannot follow a Student distribution, as no such distribution has a tail given by \eqref{que} while also satisfying this inversion symmetry.
    \item Figure \ref{verif2} presents encouraging results. Let's explain it. Since we do not know the distribution of $ Z $, we will refer to the \textbf{theoretical density} as that of $ \sqrt{N}(\g^N_{\vb{E}}(z) - \overline{z}) $, where $ \vb{E} $ is a Ginibre matrix. This density is calculated numerically using a large sample of Ginibre matrices. Next, we consider a model of non-Hermitian rotationnaly invariant matrices, denoted $ \vb{M} $, and examine the distribution $Z$ defined in Eq. \eqref{defZ}. For this, we first need to numerically estimate $ \g_{\vb{M}}(z) $ and $ \rho_{\vb{M}}(z) $, which can be done easily by considering a single realization of a large matrix $ \vb{M} $. We then hope that the distribution of $ Z$ will be predicted by the theoretical density, which seems to be supported by numerical simulations.
  
\end{enumerate}

\begin{figure}[H]
    \centering
    \begin{subfigure}{0.47\textwidth}
        \centering
        \includegraphics[width=\textwidth]{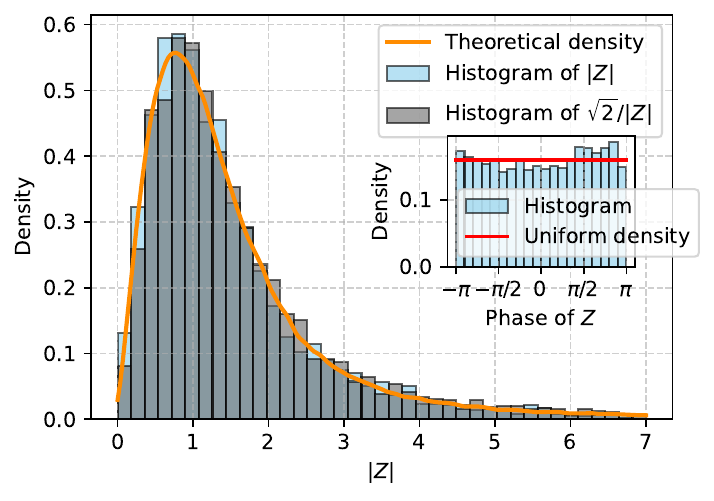}
        \caption{Distribution of $ Z $.}
    \end{subfigure}
    \hfill
    \begin{subfigure}{0.47\textwidth}
        \centering
        \includegraphics[width=\textwidth]{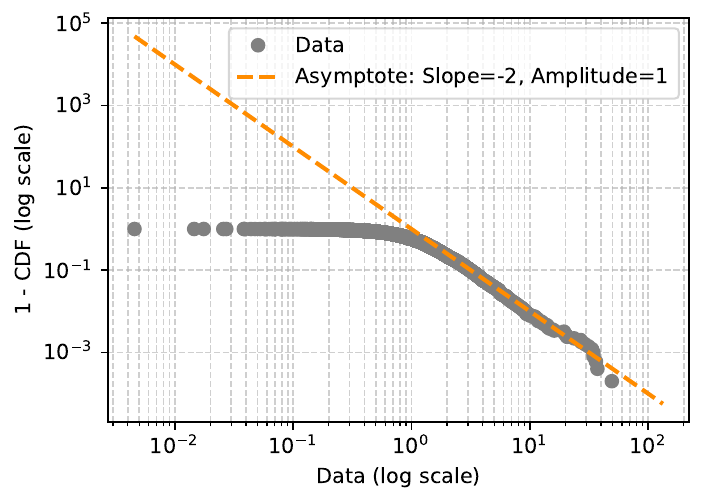}
        \caption{Log-log plot of the tail of $ Z $'s distribution.}
    \end{subfigure}
    \caption{We choose a non-trivial example by taking a matrix $ \vb{M} = \vb{A}\vb{B}\vb{C} $ of size $ 300 $, where $ \vb{A} $ is a GinUE matrix with parameter $ \tau = 0.5 $, $ \vb{B} $ is a unitary matrix randomly chosen according to the Haar measure, and $ \vb{C} $ is a complex Ginibre matrix. We consider a set of $ 5000 $ samples of $ \vb{M} $. Let $ z = 0.2i + 0.1 $. We study $ Z := \sqrt{N}({\g_{\mathbf{M}}^N(z) - \alpha})/{\sqrt{\beta}} $, where $ \alpha $ and $ \beta \approx 2.4$ are estimated as described in remark 2., where the theoretical density is also defined.
}
    \label{verif2}
\end{figure}

\newpage
\section{Summary}

We conclude the results section of this paper by presenting a summary table that provides an overview of both established results for the symmetric case and conjectures for the non-Hermitian case.

\begin{table}[h!]
    \centering
    \setlength{\tabcolsep}{10pt}
    \renewcommand{\arraystretch}{2}
    \begin{tabular}{|c|c|c|}
        \hline
        & \textbf{ \makecell{Diagonal coefficient\\of resolvent}} & \textbf{\makecell{Normalized trace \\of resolvent}} \\ 
        \hline
        \makecell{\textbf{Symmetric Case}\\$\vb{A}$: a rotationally invariant matrix\\
        $x\in \R$ inside the limiting spectrum\\
        $\vb{G}_{\vb{A}}^N(x)$: the matrix $(x1 - \vb{A})^{-1}$ \\
        $\g_{\vb{A}}^N$: the normalized trace of $\vb{G}_{\vb{A}}^N$
        \\        
        $\rho_{\vb{A}}(x)$: the limiting spectral distribution\\
        $\h_{\vb{A}}(x)$: the Hilbert-transform of $\rho_{\vb{A}}(x)$\\                
        }        
         & 
        \makecell[l]{\\[1.5ex]
        $\displaystyle \frac{[\vb{G}_{\vb{A}}^N(x)]_{11} - \h_{\vb{A}}(x)}{\pi \rho_{\vb{A}}(x)}\xrightarrow{\mathcal{L}} X$ \\[1.5ex] 
        whose p.d.f is $\displaystyle \frac{1}{\pi} \frac{1}{1 + x^2}$ \\[1.5ex] 
        \textbf{Symmetry:} $X \sim \frac{1}{X}$} 
        & 
        \makecell[l]{\\[1.5ex]
        $\displaystyle \frac{\g_{\vb{A}}^N(x) - \h_{\vb{A}}(x)}{\pi \rho_{\vb{A}}(x)}\xrightarrow{\mathcal{L}} X$ \\[1.5ex]
        whose p.d.f is $\displaystyle \frac{1}{\pi} \frac{1}{1 + x^2}$ \\[1.5ex]
        \textbf{Symmetry:} $X \sim \frac{1}{X}$} 
        \\ 
        \hline
        \makecell{\textbf{Non-Hermitian case}\\$\vb{M}$: a rotationally invariant matrix\\
        $z\in \C$ inside the limiting spectrum\\
        $\vb{G}_{\vb{M}}^N(x)$: the matrix $(z1 - \vb{M})^{-1}$\\
        $\g_{\vb{M}}^N$: the normalized trace of $\vb{G}_{\vb{M}}^N$\\
        $\g_{\vb{M}}$: the limit of $\g_{\vb{M}}^N$\\
        $\rho_{\vb{M}}(z)$: the limiting spectral distribution\\  
        $\Tilde{O}_{\vb{M}}(z)$: the rescaled self-overlap
        }
         & 
        \makecell[l]{\\[1.5ex]
        $\displaystyle \frac{[\vb{G}_{\vb{M}}^N(z)]_{11} - \g_{\vb{M}}(z)}{\sqrt{\pi \tilde{O}_{\vb{M}}(z)}} \xrightarrow{\mathcal{L}} \Omega$ \\[1.5ex]
        whose p.d.f is $\displaystyle \frac{1}{\pi} \frac{1}{(1 + |\omega|^2)^2}$ \\[1.5ex]
        \textbf{Symmetry:} $\Omega \sim \frac{1}{\Omega}$} 
        & 
        \makecell[l]{\\[1.5ex]
        $\displaystyle \sqrt{N} \frac{\g_{\vb{M}}^N(z) - \g_{\vb{M}}(z)}{\sqrt{\pi \rho_{\vb{M}}(z)}} \xrightarrow{\mathcal{L}} Z$ \\[1.5ex]
        $P(|Z|\geq r) \underset{r\to +\infty}{\sim } \frac{1}{r^2}$
        \\[1.5ex]        
        \textbf{Symmetry:} $Z \sim \frac{\sqrt{2}}{Z}$} 
        \\ 
        \hline
    \end{tabular}
    \caption{Summary table.}
    \label{tab:recap}
\end{table}

The first row highlights known results, previously introduced and discussed in the introductory section. The second row presents conjectures specific to the non-Hermitian case. Among these, only the distribution of the diagonal coefficient of the resolvent where $ \vb{M} $ is a Ginibre matrix has been rigorously demonstrated within this work. The remaining conjectures still need to be formally proven, and the distribution of $Z$ also remains to be found, offering promising directions for further research, in particular concerning its universality. 

\subsection*{Acknowledgments} We want to warmly thank Yan Fyodorov for discussions and encouragements. 
This research was conducted within the Econophysics
\& Complex Systems Research Chair, under the aegis of
the Fondation du Risque, the Fondation de l’Ecole polytechnique, the Ecole polytechnique and Capital Fund
Management.

\newpage

\section{APPENDIX}

\subsection{Reminders}
The characteristic function of a random complex variable $ Z $ is defined as

\begin{equation}
    \phi_X(\omega) = \mathbb{E}\left(\exp(i \Re(\overline{\omega} Z))\right)
\end{equation}
with $ \omega \in \mathbb{C} $.

According to Levy's theorem for complex random variables, if $ Z_n $ is a sequence of complex-valued random variables converging in distribution to a complex random variable $ Z $, then this is equivalent to the pointwise convergence of their characteristic functions, i.e., if and only if $ \phi_{Z_n}(t) $ converges to $ \phi_Z(t) $ for all $ t $.

We recall that the probability density function for an $ N $-dimensional complex Gaussian variable $ \vb{Z} $ is

\begin{equation}
    \frac{1}{\pi^N \det(\mathbf{C})} \exp\left(-(\boldsymbol{\psi} - \boldsymbol{\mu})^* \mathbf{C}^{-1} (\boldsymbol{\psi} - \boldsymbol{\mu})\right),
\end{equation}
where $ \boldsymbol{\mu} \in \mathbb{C}^N $ is the mean vector and $ \mathbf{C} $ is the covariance matrix.

The following formula will be useful. If $ \mathbf{C} $ is a Hermitian positive definite matrix and $ \boldsymbol{\mu} \in \mathbb{C}^N $, then

\begin{equation}\label{formuladet}
    \frac{1}{\pi^N} \int_{\boldsymbol{\psi} \in \mathbb{C}^N} \exp\left(-(\boldsymbol{\psi} - \boldsymbol{\mu})^* \mathbf{C}^{-1} (\boldsymbol{\psi} - \boldsymbol{\mu})\right) d\boldsymbol{\psi} = \det(\mathbf{C}),
\end{equation}
where

\begin{equation}
    d\boldsymbol{\psi} := \bigwedge_{k=1}^N d \Re(\boldsymbol{\psi}_k) \wedge d \Im(\boldsymbol{\psi}_k) = \bigwedge_{k=1}^N \frac{i \, d \boldsymbol{\psi}_k \wedge d \overline{\boldsymbol{\psi}_k}}{2}.
\end{equation}

For a one-dimensional complex Gaussian variable $ Z $ with mean $ \mu $ and variance $ \sigma^2 $, the probability density function can be expressed as

\begin{equation}
    \frac{1}{\pi \sigma^2} \exp\left(-\frac{|z - \mu|^2}{\sigma^2}\right).
\end{equation}
We can represent $ Z $ as $ X + iY $, where $ X $ and $ Y $ are two independent real Gaussian variables. $ X $ and $ Y $ have means $ \Re(\mu) $ and $ \Im(\mu) $, respectively, and both have a variance equal to $ \frac{\sigma^2}{2} $. It is a well-known fact that the characteristic function of $ Z $ is

\begin{equation}
    \phi_Z(\omega) = \exp\left(-\frac{|\omega - \mu|^2 \sigma^2}{4}\right)
\end{equation}
with $ \omega \in \mathbb{C} $.

\subsection{Proof of theorem \ref{theorem1}}
We use the Schur complement formula to compute the $(1, 1)$ element of $[\vb{G}_{\vb{M}}(z)]_{11}$ as done on page 21 in \cite{potters2020first}:

\begin{equation}
    \frac{1}{[\vb{G}_{\vb{M}}]_{11}} = z - \vb{M}_{1 1}- \vb{A}
\end{equation}
with
\begin{equation}
    \vb{A} = \sum_{k, l = 2}^N \vb{M}_{1 k} [(z\vb{1} - \vb{F})^{-1}]_{kl} \vb{M}_{l 1}
\end{equation}
where the matrix $\vb{F}$ is the $(N-1)\times (N-1)$ submatrix of $\vb{M}$ with the first row and column removed.

Using the independence of the entries of $\vb{M}$, we can write

\begin{align}
    \phi_{\frac{1}{[\vb{G}_{\vb{M}}]_{11}}}^N(\omega) =& \E\left[\exp\left(i\Re\left(\overline{\omega}\left(z - \vb{M}_{1 1}- \vb{A}\right)\right)\right)\right] \\
    =& \exp\left(i \Re(\overline{\omega} z)\right) \exp\left(-\frac{|\omega|^2}{4N}\right)\phi_{\vb{A}}^N(\omega).\label{finiteN}
\end{align}
Given $\vb{F}, (\vb{M}_{l, 1})_{l\geq 2}$, $\vb{A}$ is just a sum of the independent Gaussian variables $\vb{M}_{1, k}$. Thus, 

\begin{align}\label{conditionning}
    \E\left[\exp\left(i\Re(\overline{\omega} \vb{A})\right) \big| \vb{F}, (\vb{M}_{l, 1})_{l\geq 2}\right] =& \exp\left(-\frac{|\omega|^2}{4N} \sum_{k=2}^N \left\lvert \sum_{l=2}^N [(z\vb{1} - \vb{F})^{-1}]_{kl} \vb{M}_{l 1}\right\rvert^2 \right)\\
    =& \exp\left(-\frac{|\omega|^2}{4N} \sum_{k=2}^N \left\lvert [\vb{E}\vb{x}]_k\right\rvert^2 \right)
\end{align}
where $\vb{x}$ is the $N-1$-dimensional vector $(\vb{M}_{l, 1})_{l\geq 2}$ and $\vb{E}$ denotes the $(N-1)\times (N-1)$ matrix $(z\vb{1} - \vb{F})^{-1}$.
Then, 

\begin{align}
    \phi_{\vb{A}}^N(\omega) &= \E\left(\exp\left(i\Re(\overline{\omega} \vb{A})\right)\right)\\
    &= \E\left[\E\left[\exp\left(i\Re(\overline{\omega} \vb{A})\right) \big| \vb{F}, (\vb{M}_{l, 1})_{l\geq 2}\right]\right]\\
    &= \E\left[ \exp\left(-\frac{|\omega|^2}{4N} \sum_{k=2}^N \left\lvert [\vb{E}\vb{x}]_k\right\rvert^2 \right)\right]\\
    &= \E\left(\exp\left(-\frac{|\omega|^2}{4N} (\vb{E}\vb{x})^* (\vb{E}\vb{x})\right)\right)\\
    &= \E\left[\exp\left(-\frac{|\omega|^2}{4N} \vb{x}^* \vb{E}^*\vb{E}\vb{x}\right)\right].
\end{align}
Since $\vb{x}$ is an $(N-1)$-dimensional complex Gaussian vector with mean $0$ and variance $\frac{1}{N}$, we can write the previous equation as 

\begin{align}
    \phi_{\vb{A}}^N(\omega) =& \frac{N^{N-1}}{\pi^{N-1}} \int_{\vb{x}\in \C^{N-1}}\exp\left(-\frac{|\omega|^2}{4N} \vb{x}^* \vb{E}^*\vb{E}\vb{x}\right)\exp\left(-N\vb{x}^*I_{N-1}\vb{x}\right)d\vb{x}
\end{align}
and by using Eq. \eqref{formuladet}, we get

\begin{equation}
    \phi_{\vb{A}}(\omega) = N^{N-1}\E\left[\frac{1}{\det\left(NI_{N-1} + \frac{|\omega|^2}{4N}\vb{E}^*\vb{E}\right)}\right].
\end{equation}
By denoting $\vb{B} = (z\vb{1} - \vb{F}) = \vb{E}^{-1}$, we have

\begin{align}
    \phi_{\vb{A}}^N(\omega) &= N^{N-1}\E\left[\frac{1}{\det\left(N + \frac{|\omega|^2}{4N}(\vb{B}^{-1})^*\vb{B}^{-1}\right)}\right]\\
    &= N^{N-1}\E\left[\frac{1}{\det\left(N + \frac{|\omega|^2}{4N}(\vb{B}\vb{B}^*)^{-1}\right)}\right]\\
    &= N^{N-1} \E\left[\frac{\det(\vb{B}\vb{B}^*)}{\det\left(N\vb{B}\vb{B}^* + \frac{|\omega|^2}{4N}\right)}\right]\\
    &= \E\left[\frac{\det(\vb{B}\vb{B}^*)}{\det\left(\vb{B}\vb{B}^* + \frac{|\omega|^2}{4N^2}\right)}\right]\\
    &= \E\left[\frac{\det{(zI_{N-1} - \vb{F})(zI_{N-1} - \vb{F})^*}}{\det \left((zI_{N-1} - \vb{F})(zI_{N-1} - \vb{F})^*\right) + \frac{|\omega|^2}{4N^2}}\right].\label{back}
\end{align}

Now, we introduce the quantity

\begin{equation}
    \mathcal{D}_{N}^{(L)}(z, p) = \E\left[\frac{\det^{L} {(zI_{N} - \vb{P})(zI_N - \vb{P})^*}}{\det \left(pI_N + (zI_N - \vb{P})(zI_N - \vb{P})^*\right)}\right]_{\vb{P}}
\end{equation}
where $L \in \mathbb{N}$, $p \geq 0$, and the expected value is taken over $\vb{P}$, an $N\times N$ matrix with i.i.d. complex Gaussian entries with mean $0$ and variance $1$.
This quantity was considered in \cite{grela2016exact}, and the exact expression was found in \cite{fyodorov2018statistics} for $L = 1$:

\begin{multline}
    \mathcal{D}_{N}^{(1)}(z, p) = \frac{1}{(N-1)!}e^{|z|^2} \int_0^{+\infty} \frac{dt}{t(1+t)}e^{-pt} e^{\frac{-|z|^2}{\left(1 + \frac{1}{t}\right)}} \left(\frac{t}{1 + t}\right)^N\\
    \left(\Gamma(N+1, |z|^2) - \frac{t}{1+t}|z|^2 \Gamma(N, |z|^2)\right)
\end{multline}
where $\Gamma$ is the incomplete Gamma function defined as 

\begin{equation}
    \Gamma(s, x) = \int_x^{+\infty} t^{s-1}e^{-t}dt.
\end{equation}

Returning to Eq.\eqref{back}, we get

\begin{equation}
    \phi_{\vb{A}}^N(\omega) = \frac{1}{N^N} \mathcal{D}_{N-1}^{(1)}\left(\sqrt{N}z, \frac{|\omega|^2}{4N}\right)
\end{equation}
and so,

\begin{multline}
    \phi_{\vb{A}}^N(\omega) = \frac{1}{(N-2)!}e^{N|z|^2} \int_0^{+\infty} \frac{dt}{t(1+t)}e^{-\frac{|\omega|^2 t}{4N}} e^{\frac{-N|z|^2}{\left(1 + \frac{1}{t}\right)}} \left(\frac{t}{1 + t}\right)^{N-1}\\
    \left(\Gamma(N, N|z|^2) - \frac{tN}{1+t}|z|^2 \Gamma(N-1, N|z|^2)\right)dt.
\end{multline}
Inserting the previous formula in Eq.\eqref{finiteN} yields

\begin{multline}    
    \phi_{\frac{1}{[\vb{G}_{\vb{M}}]_{11}}}^N(\omega) = \exp\left(i \Re(\overline{\omega} z)\right) \exp\left(-\frac{|\omega|^2}{4N}\right)\frac{1}{(N-2)!}e^{N|z|^2} \int_0^{+\infty} \frac{dt}{t(1+t)}e^{-\frac{|\omega|^2 t}{4N}} e^{\frac{-N|z|^2}{\left(1 + \frac{1}{t}\right)}} \left(\frac{t}{1 + t}\right)^{N-1}\\
    \left(\Gamma(N, N|z|^2) - \frac{tN}{1+t}|z|^2 \Gamma(N-1, N|z|^2)\right)dt.
\end{multline}
Using the following recurrence relation satisfied by the incomplete gamma function,
\begin{equation}
    \Gamma(s + 1, x) = s\Gamma(s, x) + x^s e^{-x}
\end{equation}
and recalling the characteristic function of a complex Gaussian variable, we then complete the proof of Theorem \ref{theorem1}.

\subsection{Proof of Theorem \ref{main}}

We continue to consider the variable $ t $ from Theorem \ref{theorem1}.

\subsubsection{Regime (1)}

Let $ h $ be a bounded continuous function. Equation \ref{sigma} yields

\begin{multline}
    \E\left[h\left(\frac{t}{N g(N)}\right)\right] = \frac{1}{(N-2)!}\int_0^{+\infty}\frac{h\left(\frac{t}{N g(N)}\right)dt}{t(1+t)} e^{\frac{-N(1 - g(N))}{\left(1 + \frac{1}{t}\right)}} \left(\frac{t}{1 + t}\right)^{N-1}
    \\\left(e^{N(1-g(N))}\Gamma(N-1, N(1 - g(N)))\left(N-1 - \frac{tN}{1+t}(1 - g(N))\right)  +\left(N(1 - g(N))\right)^{N-1}\right).
\end{multline}
And so, by using a change of variables \( u = \frac{t}{N g(N)} \), the integral becomes:

\begin{multline}
    \E\left[h\left(\frac{t}{N g(N)}\right)\right] = \frac{1}{(N-2)!}\int_0^{+\infty}\frac{h(u)du}{u(1+Ng(N)u)} e^{\frac{-N(1 - g(N))}{\left(1 + \frac{1}{Ng(N)u}\right)}} \left(\frac{Ng(N)u}{1 + Ng(N)u}\right)^{N-1}
    \\\left(e^{N(1-g(N))}\Gamma(N-1, N(1 - g(N)))\left(N-1 - \frac{uN^2 g(N)}{1+Ng(N)u}(1 - g(N))\right)  +\left(N(1 - g(N))\right)^{N-1}\right).
\end{multline}
Using an asymptotic result on the $\Gamma$ function, which is
\begin{equation}
   \frac{\Gamma(N-1, N(1 - g(N)))}{(N-2)!} \underset{N\to \infty}{\longrightarrow} 1, 
\end{equation}
when $ \frac{1}{N} \ll g(N) $ as well as the limit

\begin{equation}
    \left(\frac{Ng(N)u}{1 + Ng(N)u}\right)^{N-1}e^{N(1 - g(N))}  e^{\frac{-N(1 - g(N))}{\left(1 + \frac{1}{Ng(N)u}\right)}} \underset{N\to \infty}{\longrightarrow} e^{-\frac{1}{u}},
\end{equation}
the dominated convergence theorem allows us to conclude that the density function of $ \frac{t}{Ng(N)} $ will converge to

\begin{equation}\label{density}
   u\mapsto\frac{1}{u^2}e^{-\frac{1}{u}}.
\end{equation}

Then, Theorem \ref{theorem1} tells us that $ \frac{1}{\sqrt{g(N)}}\left(\frac{1}{[\vb{G}_{\vb{M}}]_{11}} - z\right) $ will converge in distribution to the random variable $ \sqrt{V} X $, where $ V $ is an inverse Gamma-distributed variable whose probability density function is given by Eq.\eqref{density} and $ X $ is a standard complex Gaussian variable independent of $ V $. In conclusion, Lemma \ref{gamma} shows that the distribution of $ \sqrt{V} X $ is the complex Student variable distributed as 

\begin{equation}\label{den}
    \frac{1}{\pi}\frac{1}{(1 + |\omega|^2)^2}.
\end{equation}
Hence the result \eqref{regime2} using Lemma \ref{inverse_distrib}.

Let's look at the proof of the result \eqref{student}. If $ g(N) = 1 - |z|^2 $ is fixed, we have just seen that $ \frac{1}{\sqrt{1 - |z|^2}}\left(\frac{1}{[\vb{G}_{\vb{M}}]_{11}} - z\right) $ converges to a Student variable whose density function is Eq.\eqref{den}. In other words, $ \frac{1}{[\vb{G}_{\vb{M}}]_{11}} - z $ converges in law to the Student variable distributed as

\begin{equation}
    \frac{1}{\pi}\frac{1 - |z|^2}{(1 - |z|^2 + |\omega|^2)^2}.
\end{equation}
Lemma \ref{inverse} then allows us to determine the distribution of $ [\vb{G}_{\vb{M}}]_{11} $ as given by Theorem \ref{theorem1}.

\subsubsection{Regime (2)}

Let $ h $ be a bounded continuous function. Using \eqref{sigma}, we have

\begin{multline}
    \E\left[h\left(\frac{t}{\sqrt{N}}\right)\right] =\frac{1}{(N-2)!}\int_0^{+\infty} \frac{h(u)du}{u(1+u\sqrt{N})}e^{\frac{-N(1 + \epsilon(N))}{\left(1 + \frac{1}{\sqrt{N}t}\right)}} \left(\frac{u\sqrt{N}}{1 + u\sqrt{N}}\right)^{N-1}\\\left(e^{N(1 + \epsilon(N))} \Gamma(N-1, N(1 + \epsilon(N)))\left(N-1 - \frac{uN^{3/2}}{1+u\sqrt{N}}(1 + \epsilon(N))\right)  +\left(N(1 +\epsilon(N))\right)^{N-1}\right).
\end{multline}
Classical asymptotic results such as

\begin{equation}
    \sqrt{N}(N-2)! \sim \sqrt{2\pi N}\left(\frac{N-2}{e}\right)^{N-2}\sqrt{N} \sim \sqrt{2\pi} e^{-N} N^{N-1}
\end{equation}
and

\begin{equation}
    \sqrt{N}\Gamma(N-1, N(1 + \epsilon(N))\sim \frac{\sqrt{2\pi}}{2} e^{-N}N^{N-1}
\end{equation}
with $ \epsilon(N) \ll \frac{1}{\sqrt{N}} $ lead us to

\begin{equation}
    \E\left[h\left(\frac{t}{\sqrt{N}}\right)\right] \underset{N\to \infty}{\longrightarrow} \int_0^{+\infty} \frac{h(u)}{u^2}e^{-\frac{1}{2u^2}}\left(\frac{1}{2u} + \frac{1}{\sqrt{2\pi}}\right).
\end{equation}
Hence the result.

\subsubsection{Regime (3)}

Let $ h $ be a bounded continuous function. Equation \ref{sigma} yields

\begin{multline}
    \E\left[h\left(\frac{t}{\sqrt{N}}\right)\right] = \frac{1}{(N-2)!} \int_0^{+\infty} \frac{h(u)du}{u(1+\sqrt{N}u)} e^{\frac{-N\left(1 + \frac{\alpha}{\sqrt{N}}\right)}{\left(1 + \frac{1}{\sqrt{N}u}\right)}} \left(\frac{\sqrt{N}u}{1 + \sqrt{N}u}\right)^{N-1}\\
    \left[e^{N\left(1 + \frac{\alpha}{\sqrt{N}}\right)}\Gamma\left(N-1, N\left(1 + \frac{\alpha}{\sqrt{N}}\right)\right)\left(N-1 - \frac{uN^{3/2}}{1+u\sqrt{N}}\left(1 + \frac{\alpha}{\sqrt{N}}\right)\right)  +\left(N\left(1 + \frac{\alpha}{\sqrt{N}}\right)\right)^{N-1}\right].
\end{multline}
Using this time that

\begin{equation}
    \frac{\Gamma\left(N-1, N\left(1 + \frac{\alpha}{\sqrt{N}}\right)\right)}{(N-2)!}\underset{N\to \infty}{\longrightarrow} \frac{1}{\sqrt{2\pi}}\int_{\alpha}^{+\infty} e^{-\frac{v^2}{2}}dv,
\end{equation}
we get

\begin{equation}
    \E\left[h\left(\frac{t}{\sqrt{N}}\right)\right] \underset{N\to \infty}{\longrightarrow}  \int_0^{+\infty}\frac{h(u)}{u^2}e^{-\frac{1}{2u^2} + \frac{\alpha}{u}}\left(\frac{1 - \alpha u}{u \sqrt{2\pi}}\int_{\alpha}^{+\infty}e^{-v^2/2}dv  + \frac{e^{-\frac{\alpha^2}{2}}}{\sqrt{2\pi}}\right)du.
\end{equation}
Hence the result.

\subsubsection{Regime (4)}

Let $ h $ be a bounded continuous function. Equation \eqref{sigma} gives

\begin{multline}
    \E\left[h\left(f(N)t\right)\right] = \frac{1}{(N-2)!} \int_0^{+\infty} \frac{h(u)du}{u(1+u/f(N))} e^{\frac{-N(1 + f(N))}{\left(1 + \frac{f(N)}{u}\right)}} \left(\frac{1}{1 + f(N)/u}\right)^{N-1}\\\left(e^{N(1 + f(N))}\Gamma(N-1, N(1 + f(N)))\left(N-1 - \frac{Nu/f(N)}{1+u/f(N)}(1 + f(N))\right)  +\left(N(1 + f(N))\right)^{N-1}\right).
\end{multline}
Using Laplace's method, since $ Nf(N)^2 $ tends to infinity, we find that

\begin{equation}
    \E\left(h(f(N)t\right) \underset{N\to \infty}{\longrightarrow} h(u^*)
\end{equation}
where $ u^* $ maximizes either the function $ u\mapsto -\frac{1}{u} + \frac{1}{2u^2} $ if $ f(N)\underset{N\to \infty}{\longrightarrow} 0 $ or the function $ u\mapsto -\frac{1 + l}{1 + \frac{l}{u}} + \ln\left(1 + \frac{l}{u}\right) $ if $ f(N)\underset{N\to \infty}{\longrightarrow} l $ with $ l > 0 $. In these two cases, $ u^* = 1 $.

We can then conclude using Theorem \ref{main} (noting that the variance in Eq. \eqref{not} is $ \frac{t + 1}{N} $).

\subsection{Lemmas}

\begin{lemma}\label{inverse_distrib}
Let $(X_n)$ be a sequence of random variables, $z \in \mathbb{C}^*$, and $(u_n)$ a complex sequence. We assume that

\begin{equation}
    u_n \underset{n \to \infty}{\longrightarrow} +\infty.
\end{equation}
and

\begin{equation}
    u_n\left(\frac{1}{X_n} - z\right)\xrightarrow{\mathcal{L}} X.
\end{equation}
Then,

\begin{equation}
    z^2 u_n\left(\frac{1}{z} - X_n\right)\xrightarrow{\mathcal{L}} X.
\end{equation}
\end{lemma}

\begin{proof}
    We can write

    \begin{equation}
        z^2 u_n\left(\frac{1}{z} - X_n\right) = z X_n u_n \left( \frac{1}{X_n} - z \right).
    \end{equation}
Using Slutsky's theorem, it suffices to show that $z X_n$ converges in probability to $1$.
Let's fix $\epsilon > 0$, $\delta > 0$, and $A>0$ sufficiently large such that

    \begin{equation}
        P(|X| \geq A) \leq \delta.
    \end{equation}
Due to the convergence in distribution, for sufficiently large $n$,

    \begin{equation}
        P\left(u_n \left|\frac{1}{X_n} - z\right| \leq A\right) \leq 2\delta.
    \end{equation}
Since $u_n \to \infty$, there exists $n$ large enough such that $\epsilon u_n |z| \geq A$. Thus, for sufficiently large $n$,

    \begin{equation}
        P\left(\left|\frac{1}{zX_n} - 1\right| \geq \epsilon\right) = P\left(\left|u_n \left(\frac{1}{X_n} - z\right)\right| \geq \epsilon |u_n| |z|\right) \leq P\left(u_n \left|\frac{1}{X_n} - z\right| \leq A\right) \leq 2\delta.
    \end{equation}
Therefore, $\frac{1}{zX_n}$ converges in probability to $1$ and we obtain the result by the continuity of the inverse function.
\end{proof}

\begin{lemma}\label{gamma}
    Let $Z$ be a complex Gaussian variable with mean $0$ and variance $\sigma^2$. Let's assume that $t = \sigma^2$ is an inverse gamma variable distributed as 

\begin{equation}
     t\in \R^{+}\mapsto \frac{\beta^{\nu}}{\Gamma(\nu)} \frac{1}{t^{\nu + 1}}\exp{- \frac{\beta}{t}}
\end{equation}
with $\nu, \beta>0$. The distribution induced on $Z$ is then given by

\begin{equation}
    \frac{1}{\pi} \frac{\beta^{\nu} \nu}{(\beta + |\omega|^2)^{\nu+1}}
\end{equation}
which is a complex Student variable with $\nu$ and $\beta$ as parameters.

\end{lemma}

\begin{proof}

The distribution of $Z$ conditionning on $\sigma$ is

\begin{equation}
    f_Z(\omega) = \frac{1}{\pi \sigma^2} \exp{\frac{-|\omega|^2}{\sigma^2}}.
\end{equation}
and so the probabilty density of $Z$ is given by

\begin{equation}
    \int_0^{+\infty} \frac{\beta^{\nu}}{\Gamma(\nu)} \frac{1}{v^{\nu + 1}}\exp{- \frac{\beta}{t}}  \frac{1}{\pi t} \exp{\frac{-|\omega|^2}{t}}dt
    = \int_0^{+\infty} \frac{\beta^{\nu}}{\pi \Gamma(\nu)} \frac{1}{t^{\nu + 2}}\exp{- \frac{|\omega|^2 + \beta}{t}} dt
\end{equation}
and by making the change of variables $t\to \frac{|\omega|^2 + \beta}{t}$, we obtain 

\begin{align}
    \int_0^{+\infty} \frac{\beta^{\nu}}{\pi \Gamma(\nu)} \frac{v^{\nu}}{(\beta + |\omega|^2)^{\nu+1}} \exp{- t} dt&= \frac{\beta^{\nu}}{\pi }  \frac{1}{(\beta + |\omega|^2)^{\nu+1}} \frac{\Gamma(\nu+1)}{\Gamma(\nu)}\\
&=\frac{1}{\pi} \frac{\beta^{\nu} \nu}{(\beta + |\omega|^2)^{\nu+1}}.
\end{align}
\end{proof}

\begin{lemma}\label{inverse}
    Let $Z$ be a complex Student variable with parameters $\nu = 2$, $\beta$ and $c$ mean such that the distribution of $Z$ is

    \begin{equation}
        \frac{1}{\pi} \frac{\beta}{(\beta + |\omega - c|^2)^{2}}.
    \end{equation}
Then, $\overline{\frac{1}{Z}}$ is distributed as 

    \begin{equation}
    \frac{1}{\pi} \frac{\frac{\beta}{(\beta + |c|^2)^2}}{\left(\frac{\beta}{(\beta + |c|^2)^2} + \left\lvert z - \frac{c}{\beta + |c|^2} \right\rvert^2\right)^2}.
\end{equation}

\end{lemma}

\begin{proof}
    Let $f$ a complex-valued regular function. By recalling that $dx \wedge dy = \frac{id\omega \wedge \overline{d\omega}}{2}$, we can write

\begin{align}
    \E\left(f\left(\frac{1}{Z}\right)\right) =& \int_{\omega, \overline{\omega}} f\left(\frac{1}{\omega}\right) \frac{1}{\pi} \frac{\beta}{(\beta + |\omega - c|^2)^{2}} \frac{id\omega \wedge \overline{d\omega}}{2}.
\end{align}
We make the changes of variables $u = \frac{1}{\omega}$ and $\overline{u} = \frac{1}{\overline{\omega}}$:

\begin{align}
    \E\left(f\left(\frac{1}{Z}\right)\right) =& \int_{u, \overline{u}} f(u) \frac{1}{\pi} \frac{\beta}{\left(\beta + \left|\frac{1}{u} - c\right|^2 \right)^{2}} \frac{1}{u^2 \overline{u}^2}\frac{idu \wedge \overline{du}}{2}\\
    =&\int_{u, \overline{u}} f(u) \frac{1}{\pi} \frac{\beta}{\left(\beta |u|^2 + \left|1 - c u\right|^2 \right)^{2}} \frac{idu \wedge \overline{du}}{2}.
\end{align}
Thus, the result is obtained by writing the denominator in a more canonical form.
\end{proof}

\end{document}